\documentclass[smallextended]{svjour3}

\usepackage{amsmath}
\usepackage{amssymb}
\usepackage{mathrsfs}
\usepackage{comment}
\usepackage{color}
\usepackage{wrapfig}
\usepackage{floatrow}
\usepackage{graphicx,epic,tikz,pifont}
\usepackage{caption}
\usepackage{subcaption}

\usepackage{natbib}

\newcommand{\pp}{\mathbb{P}}
\newcommand{\ee}{\mathbb{E}}

\newcommand{\tsp}{\mathcal{T}}

\newcommand{\pf}{A}
\newcommand{\ch}{C}

\newcommand{\var}{\mathbb{V}}
\newcommand{\cov}{Cov}
\newcommand{\corr}{\rho}

\newcommand{\YHK}{y}
\newcommand{\PDA}{u}

\newcommand{\epf}{\hfill $\square$}

\newcommand{\pyule}{\mathbb{P}_{\YHK}}
\newcommand{\puni}{\mathbb{P}_{\PDA}}

\usepackage{marvosym}
\newcommand{\envelope}{(\raisebox{-.5pt}{\scalebox{1.45}{\Letter}}\kern-1.7pt)}

\usepackage[text={6.6in,9.6in},centering]{geometry}
\graphicspath{{./figs/}}

\begin{document}

\title{On joint subtree distributions under two evolutionary models
}

\author{Taoyang Wu \and Kwok Pui Choi}

\institute{
T. Wu 
\at School of Computing Sciences, University of East Anglia, Norwich, United Kingdom \\\email{taoyang.wu@uea.ac.uk, taoyang.wu@gmail.com}
\and
K.P. Choi
\at Department of Statistics and Applied Probability, and Department of Mathematics, National University of Singapore, Singapore 117546, \\\email{stackp@nus.edu.sg}
}



\date{\today}

\maketitle

\begin{abstract}
In population and evolutionary biology, hypotheses about micro-evolutionary and macro-evolutionary processes are commonly tested by comparing the shape indices of empirical evolutionary trees with those predicted by neutral  models. A key ingredient in this approach is the ability to compute and quantify distributions of various tree shape indices under random models of interest. As a step to meet this challenge,  in this paper we investigate the joint distribution of cherries and pitchforks (that is, subtrees with two and three leaves) under two widely used null models: the Yule-Harding-Kingman (YHK) model and the proportional to distinguishable arrangements (PDA) model. Based on two novel recursive formulae, we propose  a dynamic approach to numerically compute the exact joint distribution (and hence the marginal distributions) for trees of any size.
We also obtained insights into the statistical properties of trees generated under these two models,  including a constant correlation  between the cherry and the pitchfork distributions under the YHK model, the log-concavity and unimodality of cherry distributions under both models. In particular, we show the existence of a unique change point for cherry distribution between the two models,  that is,  there exists a critical value $\tau_n$ for each $n\geq 4$  such that the probability that a random tree with $n$ leaves  generated under the YHK model contains $k$ cherries is lower than that under the PDA model if $1<k< \tau_n$, and higher if $\tau_n<k\le n/2$. 

\bigskip

\end{abstract}


\newpage
\section{Introduction}

In population and evolutionary biology, hypotheses about micro-evolutionary processes (e.g. population genetics) and macro-evolutionary processes (e.g. speciation and extinction) are commonly tested by comparing frequencies of the shapes of empirical evolutionary trees with those predicted by null models~\citep[see, e.g.][]{mooers97a, nordborg01a,blum2006random,purvis2011shape,hagen2015age}. One challenge in this approach is the ability to compute the distributions of various tree shape indices under the models of interest, which is needed in statistical testing for calculating the $p$-value of the empirical shape statistics or constructing a confidential interval. Even for some relatively simple null models,  this can still be a challenging task. Many current approaches are based on  approximating techniques, such as  Monte Carlo sampling~\citep[see, e.g.][]{blum2006random} or Gaussian approximation~\citep[see, e.g.][]{McKenzie2000}, which could be computationally intensive or restricting the tests to the trees above a certain size.  Therefore it is desirable to explore efficient ways of computing these distributions exactly. 


Two widely used null models for generating random trees in population and evolutionary biology are the Yule-Harding-Kingman (YHK) model~\citep{harding71a,yule25a,Kingman1982}  and the proportional to different  arrangements (PDA) model~\citep{Aldous2001}. 
Under the PDA model all rooted binary trees of the same size are chosen with the same probability \citep{Aldous2001} whilst under the YHK model each tree is chosen with a probability proportional to the number of total orderings that can be assigned to its internal nodes so that the relative partial ordering derived from the tree topology is preserved. 

In this paper, we are interested in the exact computation of the joint distribution for the number of subtrees under the YHK and PDA model.  Here a subtree, also known as a fringe subtree in ~\citet{aldous1991asymptotic}, consists of a node and all its descendants. More specifically, we study the distributions of the number of cherries, subtrees with two leaves, and that of pitchforks, subtrees with three leaves. Note that this is equivalent to study the joint distributions of $2$-pronged  and $3$-pronged nodes as defined in~\citep{rosenberg06a}, as well as the joint distributions of clades of size two and three as defined in~\citep{zhu11a}.

We now describe the contents of the rest of this paper. 
In the next section we gather some necessary notation and background.
In particular, we present a random tree generating process for realising both the YHK and PDA models as described in~\cite{McKenzie2000}. In contrast to the splitting model that were used in several previous studies concerning the asymptotical distributions of subtrees~\citep[see, e.g.][]{chang2010limit}, the process used here is  based on iteratively attaching leaves. We therefore also collect some observations on the change of the numbers of cherries and pitchforks in a tree when an additional leaf is attached.

In Sections~\ref{sec:YHK} and~\ref{sec:PDA} we study subtree distributions under the YHK and the PDA models, respectively. 
Our main results include two novel recursive formulae on the joint distributions of cherries and pitchforks; see Theorem~\ref{thm:yule:pf} for the one under the YHK model and Theorem~\ref{thm:pda:pf} for the one under the PDA model.  These recursions enable us to develop a dynamic approach to numerically compute the joint distributions, and hence also their marginal distributions, for trees of any size. As an example, in Fig.~\ref{fig:jd} we illustrate the distributions for trees with 200 leaves, which suggests the joint distributions under both models can be well-approximated by bivariate normal distributions. 

\begin{figure}
\centering
\begin{subfigure}{.5\textwidth}
  \centering
  \includegraphics[width=1.1\linewidth]{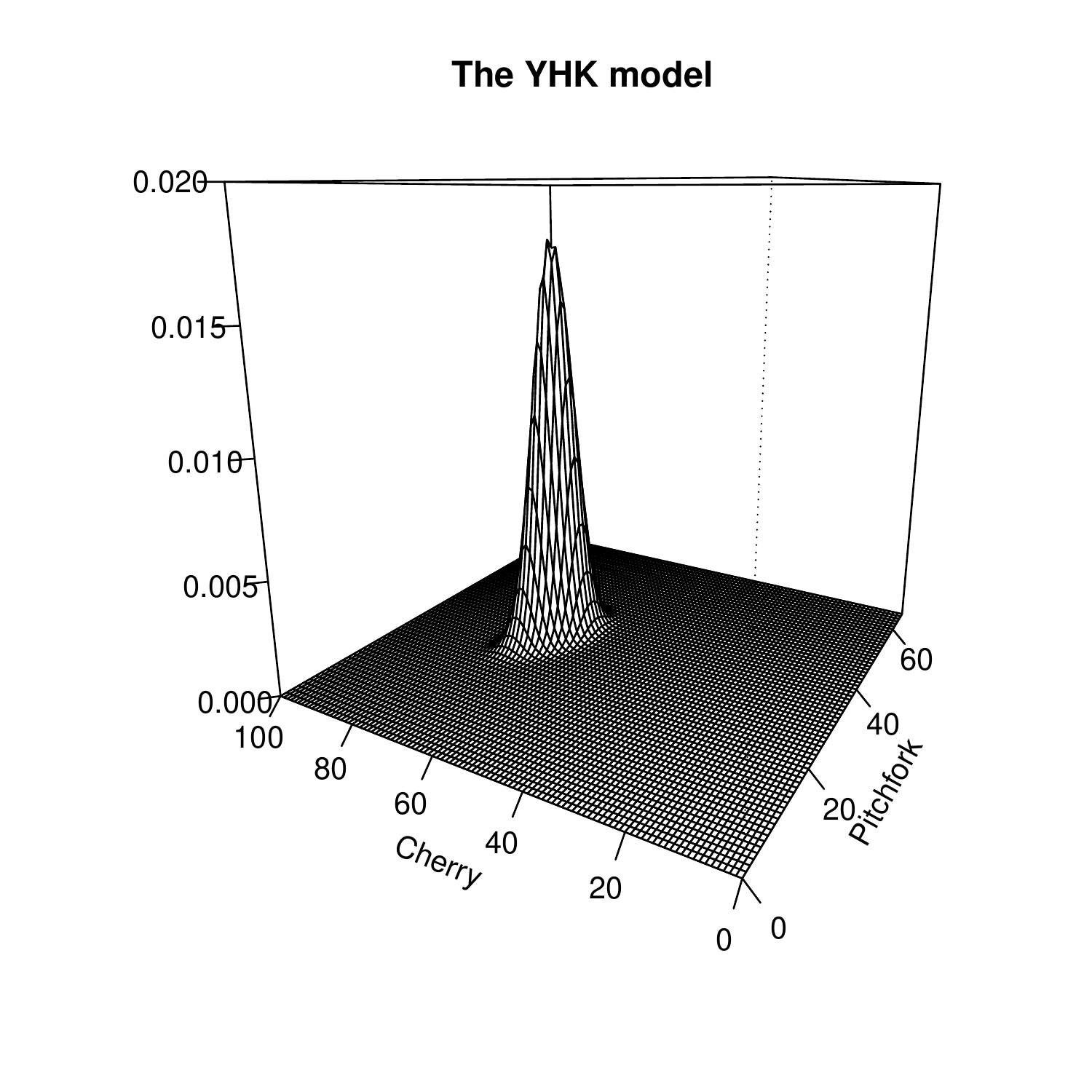}
\end{subfigure}%
\begin{subfigure}{.5\textwidth}
  \centering
  \includegraphics[width=1.1\linewidth]{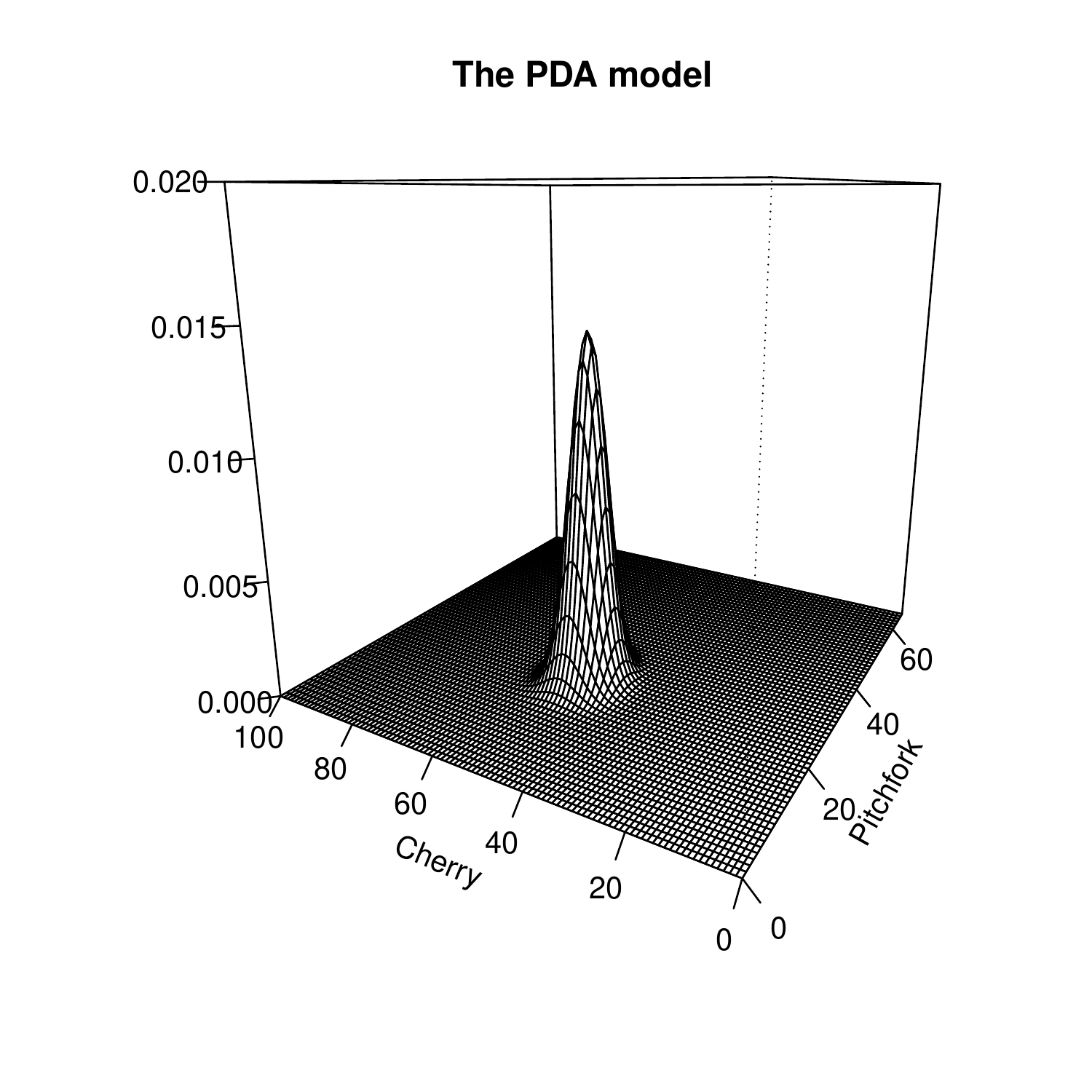}
\end{subfigure}
\caption{Probability density function for the joint distribution of cherries and pitchforks on phylogenetic trees with 200 leaves.}
\label{fig:jd}
\end{figure}

Rewritten in functional forms,  the recursions also provide a way to compute  the covariance and correlation of the joint distributions under these two models. Somewhat surprisingly, we find  that under the YHK model the correlation between the cherry and the pitchfork distributions is a constant $-\sqrt{14/69}$, which is independent of the  number of leaves (see, e.g., Corollary~\ref{cor:yule:cor}).  In addition, the recursions also lead to  an  alternative, and arguably more unified approach to computing the moments of these two distributions, and we demonstrate this by reaffirming several results obtained in previous studies. 

Using the recursions on cherry distribution derived from the joint distribution, we obtain in Theorem~\ref{ch:pda} the exact formula for cherry distribution under the PDA model, and derive some interesting properties concerning cherry distribution in general, including that this distribution is log-concave and hence unimodal under both models (see Theorems~\ref{thm:lc:yule} and~\ref{thm:lc:pda}).

In Section~\ref{sec:comp} we present a comparative study of cherry and pitchfork distributions under the YHK and PDA models.
We first compare the mean and the variance of these two distributions under these two models. Then we show in Theorem~\ref{thm:change}  that  there exists a unique change point  when comparing cherry distribution, that is,  there exists a critical value $\tau_n$ for each $n\geq 4$  such that the probability that a random tree with $n$ leaves  generated under the YHK model contains $k$ cherries is lower than that under the PDA model if $1<k< \tau_n$, and higher if $\tau_n<k\le n/2$.  
Finally, we conclude in Section~\ref{sec:discussion} with discussions and some open problems.



\newpage
\section{Preliminaries}
\label{sec:pre}

For later use, we present in this section some basic notation and results  concerning phylogenetic trees. Throughout this article, $X$ denotes a finite set with $|X|=n \geq 2$.

\bigskip
\noindent
{\bf Phylogenetic trees} 
A {\em phylogenetic tree} $T=(V(T),E(T))$ on $X$ is a rooted tree with leaf set $L(T)=X$ such that the root has one child whilst all other vertices have either zero or two children (see Fig.~\ref{fig:tree} for an example).
Note that in this paper phylogenetic trees are rooted, with their edges directed away from the root. In addition, 
for technical simplicity we assume without loss of generality that the root has one child~(also referred to as planted phylogenetic trees by~\citet{baroni2005bounding}).
Let $E^*(T)$ be the set of pendant edges in $T$, i.e., those edges incident with a leaf. Then we have $|E(T)|=2n-1$ and $|E^*(T)|=n$.

\begin{figure}[h]
\begin{center}
{\includegraphics[scale=0.8]{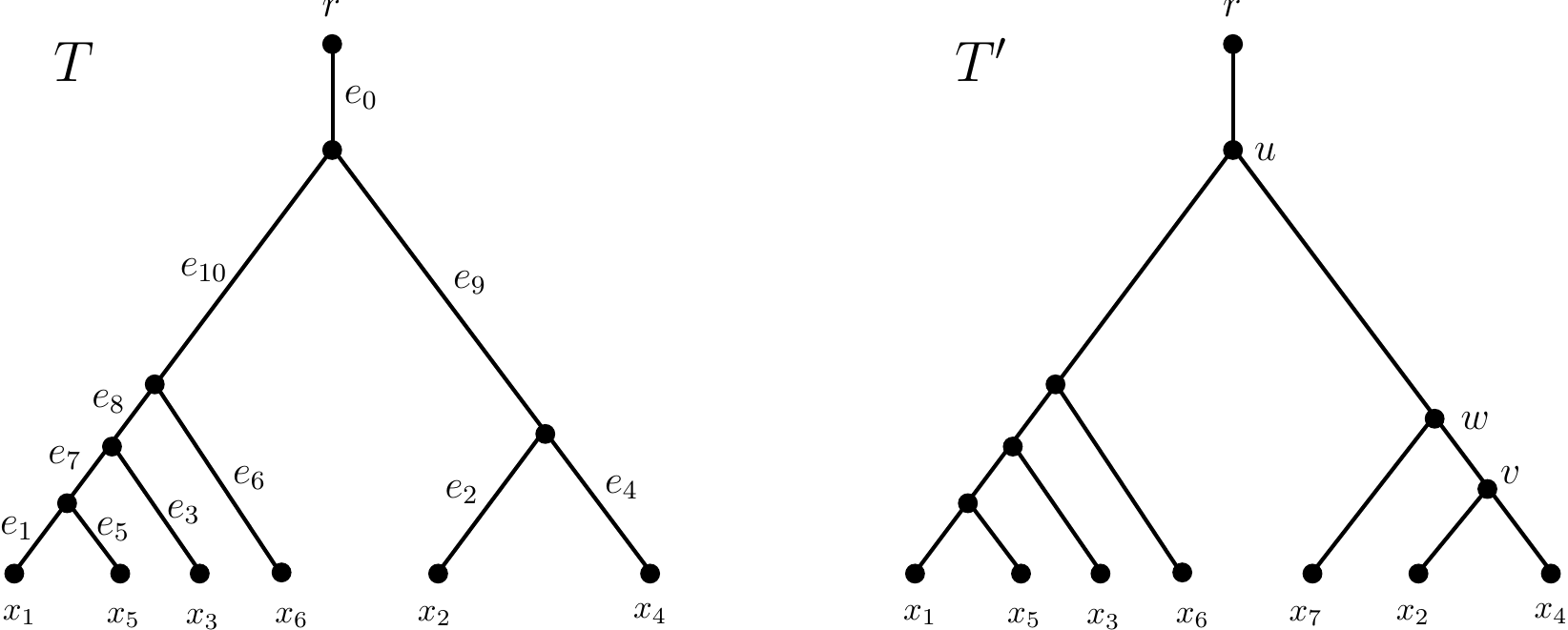}}
\end{center}
\caption{Examples of phylogenetic trees.  $T$ is a phylogenetic tree on $X=\{x_1,\dots, x_6\}$, and $T'=T[e_9;x_7]$ is a phylogenetic tree on $\{x_1,\dots,x_7\}$ that is obtained from $T$ by attaching the leaf labelled $x_7$ to edge $e_9$. Here the directions of all edges are directed away from the root $r$, and hence omitted for simplicity.
}
\label{fig:tree}
\end{figure}

Let $e$ be an edge in a phylogenetic tree $T$. The tree consisting of $e$ and all edges below $e$ is called a {\em subtree}  of $T$, and is denoted by $T(e)$. In particular,  a {\em cherry} is a subtree with two leaves, and a {\em pitchfork} is a subtree with three leaves.
 The number of cherries and pitchforks contained in $T$ are  denoted by $\ch(T)$ and $\pf(T)$, respectively.
Note first that we always have $1\leq \ch(T)\leq n/2$ and $0\leq \pf(T)\leq n/3$. Moreover, in our definition a cherry contains three edges and a pitchfork contains five edges. As an example, for the tree $T$ depicted in Fig.~\ref{fig:tree} we have $C(T)=2$ and $A(T)=1$. In addition, $T(e_8)$ is a pitchfork with edge set $\{e_1,e_3,e_5,e_7,e_8\}$, and $T(e_7)$ is a cherry with edge set  $\{e_1,e_5,e_7\}$. Finally, $\ch(T)$ and $\pf(T)$ are respectively equal to the number of $2$-pronged nodes and $3$-pronged nodes contained in $T$ (see~\cite{rosenberg06a} for the definitions of $r$-pronged nodes).

Given an edge $e$ in a phylogenetic tree $T$ and a taxon $x_0\not \in L(T)$, let $T[e;x_0]$ be the phylogenetic tree obtained from $T$ by attaching a new leaf labelled with $x_0$ to the edge $e$. Formally, let $e=\{u,v\}$ and let $w$ be a vertex not contained in $V(T)$, then $T[e;x_0]$ has vertex set $V(T)\cup \{x_0,w\}$ and edge set $\big(E(T)\setminus \{e\} \big) \cup \{(u,w),(v,w),(w,x_0)\}$ (see Fig.~\ref{fig:tree} for an illustration of this construction).
When the labelling of the new leaf is clear from the context, $T[e;x_0]$ is abbreviated to $T[e]$.

\bigskip
\noindent
{\bf The YHK and the PDA model}
In this subsection, we present a formal definition of the two null models investigated in this paper:  the {\it proportional to distinguishable arrangements} (PDA) model and the {\it Yule--Harding--Kingman} (YHK) model. In contrast to the splitting process used by~\cite{Aldous2001} to accommodate the two models, the random process used  here is based on iteratively attaching leaves.


%

Under the Yule--Harding model~\citep{harding71a, yule25a}, 
a rooted phylogenetic tree on $X$ is generated  as follows.
 Beginning with the tree with two leaves, we ``grow'' it by repeatedly uniformly sampling a pendant edge $e$ in the current tree $T_{cur}$ 
  and replace $T_{cur}$ by $T_{cur}[e]$. 
This process continues until a binary tree with $n$ leaves is obtained. Finally, we label each of its leaves with a label sampled randomly uniformly (without replacement) from $\{x_1,\dots,x_n\}$.
 When branch lengths are ignored, the Yule--Harding model is shown by~\citet{aldous96a} to be equivalent to the trees generated by the coalescent process in population genetics introduced by \citet{Kingman1982}, and so we call it the YHK model. 
The probability of generating a tree $T$ under this model is denoted by $\pyule(T)$.

Let $\tsp_n$ be the set of phylogenetic trees with leaf set $\{x_1,\dots,x_n\}$. It is well known that the number of trees contained in $
\tsp_n$ is $\varphi(n):=(2n-3)!! = 1\times 3 \times \dotsb \times (2n-3)$~\citep[see e.g.][]{semple03a}. Here we adopt  the convention that $\varphi(1)=1$. 
Under the PDA model, each tree has the same probability, that is, 
$1/{\varphi(n)}$, to be generated.
Alternatively, a tree can be generated under the PDA model using a Markov process similar to the one used in the YHK model; the only difference is that the edge $e$ is uniformly sampled from $E(T)$, instead of $E^*(T)$.  
We use $\ee _{\YHK}$, $\var _{\YHK}$, $\cov _{\YHK}$ and $\corr _{\YHK}$ to denote respectively  the  expectation, variance, covariance  and correlation 
taken with respect to the probability measure $\pp _{\YHK}$ under the YHK model. Similarly, $\ee _{\PDA}$, $\var _{\PDA}$, $\cov _{\PDA}$ and $\corr _{\PDA}$ are defined with respect to the probability $\pp _{\PDA}$ under the PDA model. 

For $n\geq 2$, let $\pf_n$ (resp. $\ch_n$) be the random variable 
$\pf(T)$ (resp. $\ch(T)$) for a random tree $T$ in $\tsp_n$.  In  this paper, we are interested in the joint distributions and the marginal properties of $\pf_n$ and $\ch_n$ under  the YHK and the PDA models. 

\bigskip
\noindent
{\bf Subtree Pattern} 
\label{sec:subtree}
For later use, we present in this subsection several technical results concerning the change of the numbers of cherries and pitchforks when a new leaf is attached to a phylogenetic tree.

We begin with the following notation. Given a phylogenetic tree $T$, let $E_{1}(T)$ be the set of pendant edges that are contained in a pitchfork but not a cherry;  $E_{2}(T)$  the set of edges in $T$
that are contained in a cherry but not in a pitchfork;
 $E_{3}(T)$ the set of pendant edges that are contained in neither a cherry nor a pitchfork; and $E_{4}(T)=E(T)\setminus (E_{1}(T)\cup E_{2}(T) \cup E_{4}(T))$. For instance, for the tree $T$ depicted  in Fig.~\ref{fig:tree}, we have $E_1(T)=\{e_3\}$, $E_2(T)=\{e_2,e_4,e_9\}$, $E_3(T)=\{e_6\}$ and $E_4(T)=\{e_0,e_1,e_5,e_7,e_8,e_{10}\}$. In addition, $E(T)$ can be decomposed into the disjoint union of these four sets of edges.  
  The following lemma, whose proof is straightforward and hence omitted here, shows this observation holds for all phylogenetic trees,  where $\sqcup$ denotes disjoint union.

\begin{lemma}
\label{lem:edge-set}
Suppose that $T$ is a phylogenetic tree with $n$ leaves. Then we have
\begin{equation}
\label{eq:edge:dec}
E(T)=E_{1}(T)\,\sqcup\,E_{2}(T)\,\sqcup\,E_{3}(T)\,\sqcup\,E_{4}(T).
\end{equation}
In addition, we have
 $|E_{1}(T)|=\pf(T)$,
  $|E_{2}(T)|=3(\ch(T)-\pf(T))$,
 $|E_{3}(T)|=n-\pf(T)-2\ch(T)$,
 and
 $|E_{4}(T)|=n-1+3\pf(T)-\ch(T)$.
\end{lemma}

The last lemma provides a decomposition for the set of edges in a phylogenetic tree, which is useful to the study of the PDA model. For the YHK model, we need an analogous decomposition for $E^*(T)$, the set of the pendant edges in $T$. To this end, note first that we have $E_1(T)\subseteq E^*(T)$ and $E_3(T)\subseteq E^*(T)$. In addition, let $E^*_i(T):=E_i(T)\cap E^*(T)$ be the set of pendant edges  in $E_i(T)$ for $i=2,4$. Then we have the following lemma, whose proof is straightforward and hence omitted. 

\begin{lemma}
\label{lem:pen:edge-set}
Suppose that $T$ is a phylogenetic tree with $n$ leaves. Then we have
\begin{equation}
\label{eq:pendant:edge:dec}
E^*(T)=E_{1}(T)\,\sqcup\,E^*_{2}(T)\,\sqcup\,E_{3}(T)\,\sqcup\,E^*_{4}(T).
\end{equation}
In addition, we have
  $|E^*_{2}(T)|=2(\ch(T)-\pf(T))$
 and
 $|E^*_{4}(T)|=2\pf(T)$.
\end{lemma}

We end this section with the following result relating the values $\ch(T[e])-\ch(T)$ and $\pf(T[e])-\pf(T)$ to the choice of $e$. 

\begin{proposition}
\label{prop:type}
Suppose that $e$ is an edge in a phylogenetic tree $T$ and $T'=T[e]$. Then 
we have
\small{
\begin{equation*}
\pf(T') = \begin{cases}
             \pf(T)  & \text{if } e\in E_3(T)\cup E_4(T), \\
             \pf(T)-1  & \text{if } e\in E_1(T),\\
             \pf(T)+1 & \text{if } e\in E_2(T);
       \end{cases} 
      \mbox{and} \quad
\ch(T') = \begin{cases}
             \ch(T) & \text{if } e \in E_2(T)\cup E_4(T), \\
             {} & \\
             \ch(T)+1 & \text{if } e \in E_1(T)\cup E_3(T).
       \end{cases}
\end{equation*}
}
\end{proposition}

\begin{proof}
Let $\{F_1,\dots,F_k\}$ be the set of pitchforks contained in $T$, and let $\{H_1,\dots,H_l\}$ ($l \geq k$) be the set of cherries contained in $T$. Here we may assume that indices are chosen in the way so that $H_1$ is contained in $F_1$. 
 
 Suppose first that $e=(u,v) \in E_1(T)$. Swapping the labelling of $F_i$ if necessary, we may assume that $e$ is the pendant edge contained in the pitchfork $F_1$ but not in the cherry $H_1$. Let $u_0$ be the parent of $u$, and let $u_1$ be the  child of $u$ that is distinct from $v$. In addition, let $w$ be the newly added interior vertex in $T'$. Now consider $e_0=(u_0,u)$ and $e'=(u,w)$ in $T'$. Then $T'(e_0)$ is not a pitchfork as $u_0$ has four leaves as its descendants. On the other hand, $T'(e')$ is a cherry of $T'$ that is not contained in $T$. Therefore, we have $\pf(T')=k-1$ and $\ch(T')=l+1$, as required. 
 
 By a similar argument, we can establish the proposition for the other three cases, i.e., $e\in E_i(T)$ for $2\leq i \leq 4$. Since by Lemma~\ref{lem:edge-set} these four cases cover all possible choices of $e$, the proposition follows.  
\epf
\end{proof}

One useful consequence of the last proposition is the following corollary, whose proof is straightforward and hence omitted.
\begin{corollary}
\label{cor:range}
Suppose that $e$ is an edge in a phylogenetic tree $T$ with $\pf(T)=a$ and $\ch(T)=b$. Then for the phylogenetic tree $T'=T[e]$, we have 
$$
(\pf(T'),\ch(T'))\in \{(a-1,b+1),(a+1,b),(a,b+1),(a,b)\}
$$
according to the index $i$ ($1\leq i \leq 4$) with $e\in E_i(T)$.
\end{corollary}


\section{Subtree Distributions under the YHK Model}
\label{sec:YHK}

In this section, we study the distributions of the random variables $A_n$ and $C_n$ under the YHK model. 
Our starting point is the following recursion on their joint distribution.

\begin{theorem}
\label{thm:yule:pf}
We have
\begin{equation*}
\begin{split}
\pyule(\pf_{n+1}=a, \ch_{n+1}=b) &=
\frac{2a}{n}\pyule(\pf_{n}=a, \ch_{n}=b)+\frac{(a+1)}{n}\pyule(\pf_{n} =a+1, \ch_n=b-1)\\
& \qquad +\frac{2(b-a+1)}{n}\pyule(\pf_{n}=a-1, \ch_{n} =b) +\frac{(n-a-2b+2)}{n}\pyule(\pf_{n}=a, \ch_{n}=b-1)
\end{split}
\end{equation*}
for $n>3$ and $1<b<n$. Moreover,  
$\pyule(\pf_{3}=a, \ch_{3}=b)$ equals to $1$ if $(a,b)=(1,1)$, and $0$ otherwise.
\end{theorem}

\begin{proof}
Fix $n> 3$, and let $T_2,\dots,T_n,T_{n+1}$ be a sequence of random trees  generated by the YHK process, that is, $T_2$ contains two leaves and $T_{i+1}=T_i[e_i]$ for a uniformly chosen pendant edge $e_i$ in $T_i$ for $2\leq i \leq n$. In particular, we have $|E^*(T_i)|=i$ for
$2\leq i \leq n+1$. Then 
we have
\begin{align}
\label{eq:total:yule}
\pyule&(\pf_{\scriptstyle n+1}=a, \ch_{n+1}=b) =\pp(\pf(T_{n+1})=a, \ch(T_{n+1})=b) \notag \\
&=\sum_{p,q} \pp(\pf(T_{n+1})=a, \ch(T_{n+1})=b\,|\,\pf(T_n)=p,\ch(T_n)=q) \pp(\pf(T_n)=p,\ch(T_n)=q) \notag \\
&=\sum_{p,q} \pp (\pf(T_{n+1})=a, \ch(T_{n+1})=b\,|\,\pf(T_n)=p,\ch(T_n)=q) \pp _{\YHK}(\pf_n=p,\ch_n=q),
\end{align}
where the first and second equalities follow from the law of total probability,  and the definition of random variables $A_n$ and $C_n$.

Let $e_n$ be the pendant edge in $T_n$ chosen in the above YHK process for generating $T_{n+1}$, that is, $T_{n+1}=T_n[e_n]$. Since Corollary~\ref{cor:range} implies that 
\begin{equation}
\label{eq:yule:5}
 \pp(\pf(T_{n+1})=a, \ch(T_{n+1})=b~|~\pf(T_n)=p,\ch(T_n)=q)=0 
\end{equation}
for $(p,q) \not \in\{(a,b),(a+1,b-1),(a-1,b),(a,b-1)\}$, 
it suffices to consider the following four cases  in the summation in  (\ref{eq:total:yule}): case (i): $p=a, q=b$; case (ii): $p=a+1, q=b-1$; case (iii): $p=a-1, q=b$;  and case (iv): $p=a, q=b-1$.

Firstly,  Proposition~\ref{prop:type} implies that  case (i) occurs if and only if $e_n\in E_4(T_n)\cap E^*(T_n)=E_4^*(T_n)$.  
Together with Lemma~\ref{lem:pen:edge-set}, we have
\begin{align}
\label{eq:yule:1}
 \pp(\pf(T_{n+1})=a, \ch(T_{n+1})=b~|~\pf(T_n)=a,\ch(T_n)=b) &=
 \frac{|E^*_4(T_n)|}{|E^*(T_n)|} 
=\frac{2\pf(T_n)}{n}=\frac{2a}{n}.
\end{align}


Similarly,  Proposition~\ref{prop:type} implies that case (ii) 
occurs if and only if $e_n\in E_1(T_n)\cap E^*(T_n)=E_1(T_n)$. 
Hence by Lemma~\ref{lem:edge-set} we have
\begin{align}
\label{eq:yule:2}
\pp(\pf(T_{n+1})=a, \ch(T_{n+1})=b~|~\pf(T_n)=a+1,\ch(T_n)=b-1)=
\frac{|E_1(T_n)|}{|E^*(T_n)|} 
=\frac{a+1}{n}.
\end{align}

Next,  Proposition~\ref{prop:type} implies case (iii) occurs if and only if $e_n\in E_2(T_n)\cap E^*(T_n)=E_2(T_n)$. Hence using Lemma~\ref{lem:edge-set} we have
\begin{align}
\label{eq:yule:3}
\pp(\pf(T_{n+1})=a, \ch(T_{n+1})=b~|~\pf(T_n)=a+1,\ch(T_n)=b) =
\frac{|E_2(T_n)|}{|E^*(T_n)|} 
=\frac{2(b-a-1)}{n}.
\end{align}

Finally, by  Proposition~\ref{prop:type} case (iv) 
 occurs if and only if $e_n$ is contained in $E_3(T_n)\cap E^*(T_n)=E^*_3(T_n)$. Hence by Lemma~\ref{lem:pen:edge-set} it follows that
\begin{equation}
\label{eq:yule:4}
\pp(\pf(T_{n+1})=a, \ch(T_{n+1}=b)~|~\pf(T_n)=a,\ch(T_n)=b-1)=
\frac{|E^*_3(T_n)|}{|E^*(T_n)|}=\frac{n-a-2b+2}{n}. 
\end{equation}


Now substituting Eq.~(\ref{eq:yule:1})--(\ref{eq:yule:4}) into Eq.~(\ref{eq:total:yule})
 completes the proof of the theorem.
 \epf
\end{proof}



The recursion in the last theorem can be used for a dynamic approach to numerically compute the joint distribution of  $A_n$ and $C_n$ . More precisely, let $M_m$ ($m\geq 3$) be the $(m+1)\times (m+1)$ matrix whose $(i,j)$-entry is $\pyule(\pf_{m}=i-1, \ch_{m}=j-1)$. Then $M_3$ contains a unique non-zero entry, which is at position $(2,2)$ and has a value of $1$ . Next, starting with $m=4$ and assuming that $M_{m-1}$ is already constructed, each entry in $M_{m}$ can be computed using time $O(1)$, and  hence $M_m$ can be constructed in time $O(m^2)$ with $M_{m-1}$  given. In other words, $M_n$, which specifies the joint distribution of cherry and pitchfork under the YHK model, can be computed in $O(n^3)$. Note that an alternative way of computing   the joint distribution of  $A_n$ and $C_n$  under the YHK model is proposed in~\cite{dw13}, which is based on integrating and differentiating generating functions.

For later use, we rewrite the recursion in Theorem~\ref{thm:yule:pf} in the following functional form.
\begin{theorem}
\label{thm:yule:fun}
Let $\varphi: \mathbb{R} \times \mathbb{R} \to \mathbb{R}$ be an arbitrary function. Then, under the YHK model, we have
\begin{align*}
\ee _{\YHK} \varphi(\pf_{n+1}, \ch_{n+1}) &= \frac{2}{n} \ee _{\YHK}[\pf_{n} \: \varphi (\pf_{n}, \ch_{n})] + \frac{1}{n} \ee _{\YHK}[\pf_{n} \: \varphi (\pf_{n}-1, \ch_{n}+1)]  \\
& \quad + \frac{2}{n} \ee _{\YHK}[ (\ch_{n}-\pf_n) \: \varphi (\pf_{n}+1, \ch_{n})]  + \frac{1}{n}  \ee _{\YHK}[(n-\pf_{n}-2\ch_{n}) \: \varphi (\pf_{n}, \ch_{n}+1)]
\end{align*}
for $n>2$.
\end{theorem}

\begin{proof}
 Consider the indicator function $I_{(a, b)}$ on $\mathbb{R} \times \mathbb{R}$ defined as
$$I_{(a, b)}(x, y)=
\begin{cases}
1 & \mbox{ if  $x=a$ and $y=b$,}\\
0 & \mbox{ otherwise.}
\end{cases}
$$
We multiply the equation in Theorem~\ref{thm:yule:pf} by $\varphi(a,b)$ and rewrite them as follows
{
\begin{align*}
 &\ee _{\YHK}[\varphi(\pf_{n+1}, \ch_{n+1}) I_{(a, b)}(\pf_{n+1}, \ch_{n+1})]
=  \frac{1}{n} {\big \{} 2 \ee _{\YHK}[\pf_n \varphi(\pf_{n}, \ch_{n}) I_{(a, b)}(\pf_{n}, \ch_{n})]  \\
& \quad+  \ee _{\YHK}[\pf_n \varphi(\pf_{n}-1, \ch_{n}+1) I_{(a, b)}(\pf_{n}-1, \ch_{n}+1)]
  + 2 \ee _{\YHK}[(\ch_n-\pf_n) \varphi(\pf_{n}+1, \ch_{n}) I_{(a, b)}(\pf_{n}+1, \ch_{n})] \\
 & \quad + \ee _{\YHK}[(n-\pf_n-2\ch_n) \varphi(\pf_{n}, \ch_{n}+1) I_{(a, b)}(\pf_{n}, \ch_{n}+1)] {\big \} } .
\end{align*}
}
Summing over all $a$ and $b$ completes the proof.
\epf
\end{proof}

In the remainder of this section we study cherry and pitchfork distributions using Theorem~\ref{thm:yule:fun}. We begin with  a functional recursion on cherry distribution $\ch_n$. This enables us to show that the cherry distribution is log-concave under the YHK model, and obtain an alternative approach to computing the central moments of cherry distribution. \begin{proposition}
\label{prop:yule:ch:fun}
Let $\psi: \mathbb{R}  \to \mathbb{R} $ be an arbitrary function. Then 
we have
\begin{align}
\label{eq:ch:yule:functional}
\ee _{\YHK} &\psi(\ch_{n+1}) = \frac{1}{n} \ee _{\YHK}[ 2\ch_n\:\psi (\ch_{n}) + (n-2\ch_{n}) \: \psi (\ch_{n}+1)]
\end{align}
for $n>2$. In particular, we have $\pp _{\YHK}(\ch_2=1)=1$, $\pp _{\YHK}(\ch_2=k)=0$ for $k\not =1$, and
\begin{equation}
\label{eq:cherry:yule}
\pp _{\YHK}(\ch_{n+1}=k) = \frac{2k}{n}\pp _{\YHK}(\ch_{n}=k)+\frac{n-2k+2}{n}\pp _{\YHK}(\ch_{n}=k-1)
\end{equation}
for $n>2$ and $1<k<n$.
\end{proposition}

\begin{proof}
For  $\psi: \mathbb{R}  \to \mathbb{R} $ given  in the statement of the proposition, we define $\varphi^*(x,y)=\psi(y)$, a function on $\mathbb{R}  \times \mathbb{R} $. Applying Theorem~\ref{thm:yule:fun} to the function $\varphi^*$ leads to Eq.~(\ref{eq:ch:yule:functional}).  Eq.~(\ref{eq:cherry:yule}) follows from Eq.~(\ref{eq:ch:yule:functional}) by taking $\psi(x)=I_k(x)$, where $I_k(x)$ equals $1$ if $x=k$, and $0$ otherwise.
\epf
\end{proof}

Using the last proposition, the mean and the variance of cherry distribution 
can be obtained by substituting $\psi(x)=x$ and $\psi(x)=x^2$, respectively, in the recursive equation Eq.~(\ref{eq:ch:yule:functional}) in Proposition~\ref{prop:yule:ch:fun}. 

\begin{corollary}
\label{cor:yule:ch}
~\citep{heard92a,McKenzie2000}
We have $\ee _{\YHK}(\ch_n)=n/3$ for $n> 2$ and $\var _{\YHK}(\ch_n)=2n/45$ for $n\geq 5$. 
\end{corollary}

Recall that a sequence of numbers, $\{y_1,\dots,y_m\}$, is said to be  {\em positive} if each 
number in the sequence is greater than zero. 
It is called {\em log-concave} if  
 $y_{k-1}y_{k+1}\leq y_k^2$ holds for $2\leq k \leq m-1$.  Clearly, a positive sequence $\{y_k\}_{1\leq k \leq m}$ is log-concave if and only if the sequence $\{y_{k}/y_{k+1}\}_{1\leq k \leq m-1}$ is increasing. Therefore, a log-concave sequence is necessarily {\em unimodal}, that is, there exists an index $1\leq k \leq m$ such that 
\begin{equation}
\label{def:unimodal}
y_1\leq y_2 \leq \dots \leq y_k \geq y_{k+1} \geq \cdots \geq y_m
\end{equation}
holds. 
 Finally, a non-negative integer valued
random variable $Y$ with probability mass function $\{p_k: k\geq 0\}$  is log-concave if $\{p_k\}_{k\geq 0}$ is a log-concave sequence.

To show that the probability density function of $\ch_n$ is log-concave, we need the following lemma.

\begin{lemma}
\label{lem:lc:four}
Let $z_1,z_2,z_3,z_4$ be four positive numbers with  $z_2^2\geq z_1z_3$ and $z_3^2\geq z_2z_4$. Then we have
\begin{align*}
z_2z_3\geq z_1z_4~~~~~\mbox{and}~~~~~z_1z_3+z_2z_4\geq 2z_1z_4.
\end{align*}
\end{lemma}

\begin{proof}
Since $z_i$ are positive for $1\leq i \leq 4$, from $z_2^2\geq z_1z_3$ and $z_3^2\geq z_2z_4$ if follows that
$$
\frac{z_2}{z_1}\geq \frac{z_3}{z_2} \geq \frac{z_4}{z_3}.
$$
Hence we have 
\begin{align}
\label{eq:first}
z_2z_3\geq z_1z_4,
\end{align}
which completes the proof of the first inequality in the lemma.\\

To prove the second inequality in the lemma, we consider the following two cases.

\noindent
{\bf Case 1:} $z_1\geq z_2$. Together with  $z_2^2\geq z_1z_3$, this implies $z_2\geq z_3$, and hence $z_3\geq z_4$ in view of $z_3^2\geq z_2z_4$. Therefore, we have
$$
(z_1-z_2)(z_3-z_4)\geq 0.
$$
This leads to $z_1z_3+z_2z_4\geq z_1z_4+z_2z_3\geq  2z_1z_4$,
where the last inequality follows from Eq.~(\ref{eq:first}).

\noindent
{\bf Case 2:}  $z_1< z_2$. If $z_3\leq z_4$, then we have $
(z_1-z_2)(z_3-z_4)\geq 0$, and hence $z_1z_3+z_2z_4\geq z_1z_4+z_2z_3\geq  2z_1z_4$, as required. Therefore, we may assume that $z_3>z_4$.  This implies $z_1z_3\geq z_1z_4$ and $z_2z_4\geq z_1z_4$, and hence $z_1z_3+z_2z_4\geq  2z_1z_4$, as required.
\epf
\end{proof}

Using the last lemma, we present the following theorem concerning the log-concavity of  the cherry distribution under the YHK  model.
\begin{theorem}
\label{thm:lc:yule}
Under the YHK model, we have
\begin{equation}
\label{eq:yhk:lc}
\pp _{\YHK}(\ch_n=k)^2\geq \pp _{\YHK}(\ch_n=k+1)\pp _{\YHK}(\ch_n=k-1)
\end{equation}
for $n>2$ and $1<k<n$.
\end{theorem}

\begin{proof}
For simplicity, we put $a_{n,k}:=\pp _{\YHK}(\ch_n=k)$. We prove this theorem by induction; the basic case $n=3$ is straight-forward. Now assuming that $n\ge 3$ and Eq.~(\ref{eq:yhk:lc}) holds for all $1<k<n$,  it suffices to show that 
\begin{align}
\label{eq:log-concave:pf:yhk}
 a_{n+1,k}^2\geq  a_{n+1,k-1} a_{n+1,k+1}
\end{align}
 for all $1<k\leq n$. Using the recursion described in Eq.~(\ref{eq:cherry:yule}), we have
 $$
  a_{n+1,k}^2= 4k^2a_{n,k}^2+(n+2-2k)^2a_{n,k-1}^2+4k(n+2-2k)a_{n,k}a_{n,k-1} $$
  and  $a_{n+1,k-1} a_{n+1,k+1}$ is equal to
\begin{align*}
  (2k+2)(2k-2)&a_{n,k-1}a_{n,k+1}+(2k+2)(n-2k+4)a_{n,k-2}a_{n,k+1} \\
  &+(n-2k)(2k-2)a_{n,k}a_{n,k-1}+(n-2k)(n-2k+4)a_{n,k}a_{n,k-2}.
 \end{align*}
 Therefore, by the inductive assumption and Lemma~\ref{lem:lc:four}
 we have
 \begin{align*}
  & a_{n+1,k}^2-a_{n+1,k-1} a_{n+1,k+1} \\
  &=  2[k(n-2k)+(n+2k)](a_{n,k}a_{n,k-1}-a_{n,k-2}a_{n,k+1}) +4k^2(a^2_{n,k}-a_{n,k-1}a_{n,k+1})  \\
   &\quad  +4a_{n,k-1}(a_{n,k-1}-a_{n,k-2}) +4(n+2-2k)^2(a_{n,k+1}^2-a_{n,k}a_{n,k-2})+4a_{n,k-2}(a_{n,k}-a_{n,k+1})\\
   &\geq 4a_{n,k+1}(a_{n,k-1}-a_{n,k-2})+4a_{n,k-2}(a_{n,k}-a_{n,k+1}) \\
   &=4(a_{n,k-1}a_{n,k+1}+a_{n,k}a_{n,k-2}-2a_{n,k-2}a_{n,k+1}) \geq 0,
 \end{align*}
 from which Eq.~(\ref{eq:log-concave:pf:yhk}) follows, as required.
 \epf
 \end{proof}

\bigskip
In the next result we compute the mean and the variance of pitchfork distribution $\pf_n$ under the YHK model, and calculate the covariance and correlation of $\pf_n$ and $\ch_n$. Note that the mean and the variance of $\pf_n$ was also obtained by~\citet[Theorem 4.4]{rosenberg06a}. Since the proof is similar to that of Corollary~\ref{cor:yule:ch}, we only outline the main step here. 

\begin{proposition}
\label{prop:yule:distr}
For $n\geq 7$ we have
\begin{equation}
\ee _{\YHK}(\pf_n)=\frac{n}{6},
\quad
\cov_{\YHK}(\pf_n,\ch_n)=-\frac{n}{45},
\quad
\mbox{and}
\quad
\var _{\YHK}(\pf_n)=\frac{23n}{420}.
\end{equation}

\end{proposition}

\begin{proof}
Applying  Theorem~\ref{thm:yule:fun} to $\varphi(x,y)=x$ and using Corollary~\ref{cor:yule:ch}, it follows that
\begin{align*}
\ee _{\YHK}(\pf_{n+1})&=\frac{1}{n}\ee _{\YHK}[2\pf_n^2+\pf_n(\pf_n-1)+2(\ch_n-\pf_n)(\pf_n+1)+(n-\pf_n-2\ch_n)\pf_n] \\
&=\frac{2}{3}+\frac{n-3}{n}\ee _{\YHK}(A_n)
\end{align*}
holds for $n>2$. Together with $\ee _{\YHK}(A_3)=1$, we have $\ee _{\YHK}(A_n)=n/6$ for $n\geq 4$, as required.\\

Next, applying Theorem~\ref{thm:yule:fun} to the function $\varphi(x,y)=xy$ 
shows that
\begin{align*}
\ee _{\YHK}(\pf_{n+1}\ch_{n+1}) 
=\frac{n-5}{n} \ee _{\YHK} (\pf_n\ch_n)+\frac{n-1}{n} \ee _{\YHK}(\pf_n)+\frac{2}{n}\ee _{\YHK}(\ch_n^2)
\end{align*}
holds for $n>2$. By Corollary~\ref{cor:yule:ch} and $ \ee _{\YHK}(\pf_n)=n/6$ it follows that
\begin{align*}
\cov_{\YHK}(\pf_{n+1},\ch_{n+1})
&=\ee _{\YHK}(\pf_{n+1}\ch_{n+1})-\frac{(n+1)^2}{18} \\
&=\frac{n-5}{n} \ee _{\YHK} (\pf_n\ch_n)+\frac{n-1}{6}+\frac{2(5n+2)}{45} -\frac{(n+1)^2}{18}\\
&=\frac{n-5}{n}\cov_{\YHK}(\pf_n,\ch_n)-\frac{2}{15}
\end{align*}
holds for $n\geq 5$. Solving the last recursion equation, 
we obtain  $\cov_{\YHK}(\pf_n,\ch_n) = -n/45$ for $n\geq 6$, as required.

Now the formula on $\var _{\YHK}(\pf_n)$ can be established by an argument similar to that for  $\cov_{\YHK}(\pf_n,\ch_n) $ by  applying Theorem~\ref{thm:yule:fun} to the function $\varphi(x,y)=x^2$.
\epf
\end{proof}

Interestingly, the last proposition imply that the correlation coefficient between the cherry and pitchfork distribution  under the YHK model is a negative constant for $n\ge 7$. Note that negative correlation is to be expected as the more cherries are found in a tree, the more likely that there are fewer pitchforks in that tree.

\begin{corollary}
\label{cor:yule:cor}
Under the YHK model, the correlation coefficient $\corr _{\YHK}(\pf_n,\ch_n)$ between $\pf_n$ and $\ch_n$ is $-\sqrt{14/69}$, which is independent of $n$  for $n \ge 7$.
\end{corollary}

\begin{proof}
The proposition follows directly from  Corollary~\ref{cor:yule:ch} and Proposition~\ref{prop:yule:distr}.
\epf
\end{proof}

\section{Subtree Distributions under the PDA model}
\label{sec:PDA}
In this section, we shall investigate the cherry and pitchfork distributions   under the PDA model. Similar to the study on the YHK model in Section~\ref{sec:YHK}, our starting point is the following recursion relating the joint distribution of cherries and pitchforks. 


\begin{theorem}
\label{thm:pda:pf}
We have
{
\begin{align*}
\puni&(\pf_{n+1}=a, \ch_{n+1}=b)
=\frac{n+3a-b-1}{2n-1}\puni(\pf_{n}=a,\ch_n=b)+\frac{a+1}{2n-1}\puni(\pf_{n}=a+1, \ch_n =b-1) \\
&\quad +\frac{3(b-a+1)}{2n-1}\puni(\pf_{n}=a-1, \ch_n=b)+\frac{n-a-2b+2}{2n-1}\puni(\pf_{n}=a, \ch_n=b-1)
\end{align*}
}
for $n>3$ and $1<b<n$. 
Moreover, $\puni(\pf_3=a,\ch_3=b)$ equals to $1$ if $(a,b)=(1,1)$ and $0$ otherwise. 
\end{theorem}

\begin{proof}
We give a sketch of the proof as it is similar to the proof of Theorem \ref{thm:yule:pf}. 

The only modifications needed are  the conditional   
probabilities in the four cases there. For case (i),
 by Proposition~\ref{prop:type} this case occurs if and only if $e_n\in E_4(T_n)$, and hence the conditional probability is  
$|E_4(T_n)|/|E(T_n)|= (n+3a-b-1)/(2n-1)$
by Lemma~\ref{lem:edge-set}.  Using similar arguments, for case (ii), the conditional probability is  $|E_1(T_n)|/|E(T_n)|=
(a+1)/(2n-1)$. For case (iii), the conditional probability is $|E_2(T_n)|/|E(T_n)|= 3(b-a+1)(2n-1).$ Finally, for  case (iv), the conditional probability is 
$|E_3(T_n)/|E(T_n)|= (n-a-2b+2)/(2n-1).$ The rest of the proof proceeds as in the proof of Theorem \ref{thm:yule:pf}.
\epf
\end{proof}

Using an approach similar to the remark after Theorem~\ref{thm:yule:pf}, the last theorem leads to a dynamic programming approach to compute the joint distribution of cherry and pitchfork. In addition, we present the following result which 
will enable us to study the moments of $\pf_n$ and $\ch_n$, whose proof
is similar to that of Theorem~\ref{thm:yule:fun} and hence omitted.
\begin{theorem}
\label{thm:pda:fun}
Let $\varphi: \mathbb{R}  \times \mathbb{R}  \to \mathbb{R} $ be an arbitrary function. For $n>3$ we have  
\begin{align*}
& \ee _{\PDA} \varphi(\pf_{n+1}, \ch_{n+1}) \\  = & \frac{1}{2n-1} \ee _{\PDA}[(n+3\pf_{n}-\ch_n-1) \: \varphi (\pf_{n}, \ch_{n})] + \frac{1}{2n-1} \ee _{\PDA}[\pf_{n} \: \varphi (\pf_{n}-1, \ch_{n}+1)]  \\
& \quad  + \frac{3}{2n-1} \ee _{\PDA}[ (\ch_{n}-\pf_n) \: \varphi (\pf_{n}+1, \ch_{n})]  + \frac{1}{2n-1}  \ee _{\PDA}[(n-\pf_{n}-2\ch_{n}) \: \varphi (\pf_{n}, \ch_{n}+1)].
\end{align*}
\end{theorem}


\bigskip
In the remainder of this section we shall apply Theorem~\ref{thm:pda:fun} 
to study cherry and pitchfork distributions under the PDA model. To begin with, we present the following functional recursion between cherry distribution, which will enable us to obtain the exact formula for cherry distribution and show that cherry distribution is log-concave under this model.

\begin{proposition}
\label{prop:pda:ch:fun}
Let $\psi: \mathbb{R} \to \mathbb{R}$ be an arbitrary function. Then for $n>2$ we have
\begin{align}
\label{eq:ch:pda:fun}
\ee _{\PDA} &\psi(\ch_{n+1}) = \frac{1}{2n-1} \ee _{\PDA}[(n+2\ch_n-1) \: \psi (\ch_{n})] + \frac{1}{2n-1} \ee _{\PDA}[(n-2\ch_{n}) \: \psi (\ch_{n}+1)]
\end{align}
and
\begin{equation}
\label{eq:cherry:pda}
\pp _{\PDA}(\ch_{n+1}=k) = \frac{n+2k-1}{2n-1}\pp _{\PDA}(\ch_{n}=k)+\frac{n-2k+2}{2n-1}\pp _{\PDA}(\ch_{n}=k-1), \quad { 1\leq k<n}.
\end{equation}
\end{proposition}

\begin{proof}
Let $\psi: \mathbb{R} \to \mathbb{R}$ be an arbitrary function as in the statement of the proposition. Then $\varphi^*(x,y)=\psi(y)$ is an function on $\mathbb{R} \times \mathbb{R}$. Now applying Theorem~\ref{thm:pda:fun} to the function $\varphi^*$ leads to Eq.~(\ref{eq:ch:pda:fun}).  Finally, Eq.~(\ref{eq:cherry:pda}) follows from Eq.~(\ref{eq:ch:pda:fun}) by taking $\psi(x)=I_k(x)$, where $I_k(x)$ equals $1$ if $x=k$, and $0$ otherwise.
\epf
\end{proof}

Note that the recursion presented in the last proposition enables us to study the moments of cherry distribution under the PDA model. As an example, we present below an alternative computation for the mean and the variance of $\ch_n$. Since the techniques used to solve difference equations under this model is rather different from that used under the YHK model (i.e., Corollary~\ref{cor:yule:ch}), a complete proof is included here.
Note that in the proof we will use the following well-known Faulhaber's formulae (also known as Bernoulli's formulae) concerning the sum of powers of integers ~\citep[see e.g.][]{con96}.
\begin{align*}
\sum_{i=1}^n i^2 &= \frac{n(n+1)(2n+1)}{6}, 
\hspace{.8in}
&\sum_{i=1}^n i^3 &= \frac{n^2(n+1)^2}{4}, \\
\sum_{i=1}^n i^4 &=  \frac{n(n+1)(2n+1)(3n^2 + 3n-1)}{30}, 
&\sum_{i=1}^n i^5 &= \frac{n^2 (n+1)^2 (2n^2 + 2n-1)}{12}.
\end{align*}

\begin{corollary}
\label{cor:pda:ch}
~\citep[Proposition 5]{chang2010limit}
For $n\geq 2$  we have 
$$
\ee _{\PDA}(\ch_n)=\frac{n(n-1)}{2(2n-3)} \sim \frac{n}{4}
~~~\mbox{and}~~~
\var _{\PDA}(\ch_n)=\frac{n(n-1)(n-2)(n-3)}{2(2n-3)^2(2n-5)}
 \sim \frac{n}{16}.
$$
\end{corollary}

\begin{proof}
We may assume that $n\ge 3$ in the remainder of the proof as the case $n=2$ clearly holds.  Substituting $\psi(x)=x$ in the recursive equation Eq.~(\ref{eq:ch:pda:fun}) in Proposition~\ref{prop:pda:ch:fun} leads to that
\begin{align*}
\ee _{\PDA}(\ch_{n+1})&=\frac{1}{2n-1}\ee _{\PDA}\big[ (n+2\ch_n-1)\ch_n+(n-2\ch_n)(\ch_n+1)\big] =\frac{n}{2n-1}+\frac{2n-3}{2n-1}\ee _{\PDA}(\ch_n)
\end{align*}
holds for $n>2$. Together with the initial condition $\ee _{\PDA}(\ch_2)=1$, multiplying the both sides of the last difference equation on $\ee _{\PDA}(\ch_n)$ by $2n-1$ and solving it leads to
$$(2n-3)\ee _{\PDA}(\ch_n)=1+\cdots+(n-1)=\frac{n(n-1)}{2}$$ for $n>2$, as required.

\bigskip
For simplicity, let $f(n)=(2n-3)(2n-5)$ and $g(n)=(n-1)(n^2-2n-1)$.
Then
$$
\sum_{k=1}^n g(k)=\sum_{k=1}^n (k^3-3k^2+k+1)=\frac{n(n-1)(n^2-n-4)}{4}.
$$
Next, applying Proposition~\ref{prop:pda:ch:fun} to the function $\psi(x)=x^2$ implies that
\begin{align*}
\ee _{\PDA}(\ch_{n+1}^2) &=\frac{1}{2n-1} \ee _{\PDA} \, \big [\, (n+2\ch_n-1)\ch_n^2+(n-2\ch_n)(\ch_n+1)^2 \big] \\
&=\frac{n}{2n-1}+\frac{2n-2}{2n-1}\ee _{\PDA}(\ch_n)+\frac{2n-5}{2n-1} \ee _{\PDA}(\ch^2_n) \\
&=\frac{g(n+1)}{(2n-1)(2n-3)}+\frac{2n-5}{2n-1} \ee _{\PDA}(\ch^2_n)
\end{align*}
holds for $n>2$. Now multiplying $f(n+1)$ on both sides of the above recursion leads to
$$
f(n)\ee _{\PDA}(\ch_{n}^2)-f(n-1)\ee _{\PDA}(\ch_{n-1}^2)=g(n)
$$
for $n\geq 3$. Since $\ee _{\PDA}(C_2^2)=1=-g(2)$ and $g(1)=0$, we have
$$
(2n-3)(2n-5)\ee _{\PDA}(\ch_n^2)=f(n)\ee _{\PDA}(\ch_n^2)=\sum_{k=1}^n g(k)
=\frac{n(n-1)(n^2-n-4)}{4}
$$
for $n\geq 3$, from which we have 
\begin{align}
\label{eq:ch:u:sec}
\ee _{\PDA}(\ch^2_n)=\frac{n(n-1)(n^2-n-4)}{4(2n-3)(2n-5)}
\end{align}
and hence $\var _{\YHK}(\ch_n)$ follows.
\epf
\end{proof}

Another consequence of the recursion in Proposition~\ref{prop:pda:ch:fun} 
is the following exact formula on cherry distribution for the PDA model, whose proof is a straightforward application of induction and hence omitted here.  

\begin{theorem}
\label{ch:pda}
For $n\geq 2$ and $1\leq k\leq n/2$ we have
\begin{equation}
\label{eq:chDistr:pda}
\pp _{\PDA}(\ch_n=k)=
\frac{n!(n-1)!(n-2)!2^{n-2k}}{(n-2k)!(2n-2)!k!(k-1)!}~.
\end{equation}
\end{theorem}


Interestingly, a similar formula for unrooted trees was obtained by~\citet{hendy1982branch}, that is, the probability that a random tree generated by the PDA model contains exactly $k$ cherries  is 
\begin{equation}
\label{eq:chDistr:unrooted:pda}
\frac{n!(n-2)!(n-4)!2^{n-2k}}{(n-2k)!(2n-4)!k!(k-2)!}
\end{equation}
for $2\leq k \leq n/2$~\citep[see, also][Theorem 4]{McKenzie2000}.
A direct consequence of Theorem \ref{ch:pda} is that the cherry distribution  under the PDA model is log-concave, and hence also unimodal. 

\begin{theorem}
\label{thm:lc:pda}
For $n\geq 2$ and $1<k<n$ we have
\begin{equation}
\label{eq:chDistr:pda:convex}
\pp _{\PDA}(\ch_n=k)^2\geq \pp _{\PDA}(\ch_n=k+1)\pp _{\PDA}(\ch_n=k-1).
\end{equation}
Moreover, let 
$
\Delta(n)=\frac{(n+1)(n+2)}{2(2n+1)}.
$
Then 
$$
\pp _{\PDA}(\ch_n=k-1) < \pp _{\PDA}(\ch_n=k)
~\mbox{for $1<k <  \Delta(n)$, and }
\pp _{\PDA}(\ch_n=k) > \pp _{\PDA}(\ch_n=k+1)
~\mbox{for $\Delta(n)\leq k <n/2$}.
$$
\end{theorem}

\begin{proof}
Since $\pp _{\PDA}(\ch_n=k)=0$ for $k>n/2$, the theorem clearly holds for $k\geq n/2-1$. Hence in the remainder of the proof we may assume  $k< (n-2)/2$. Now by Theorem~\ref{ch:pda} we have
\begin{align}
\frac{\pp _{\PDA}(\ch_n=k-1)}{\pp _{\PDA}(\ch_n=k)}
=\frac{4k(k-1)}{(n-2k+1)(n-2k+2)} :=g(k,n). \label{eq:ch:pda:inc}
\end{align}
Considering the function 
$g(k,n)$ defined in Eq.~(\ref{eq:ch:pda:inc}),
 then $g(k+1,n)> g(k,n)$ holds for $1<k<(n-2)/2$. This, together with Eq.~(\ref{eq:ch:pda:inc}), completes the proof of Eq.~(\ref{eq:chDistr:pda:convex}).

The second part of the theorem follows from the observation that $g(k,n)> 1$ if and only if $k\geq \Delta(n)$. 
\epf
\end{proof}

Now we apply Theorem~\ref{thm:pda:fun} to study pitchfork distribution, and the joint distribution between pitchforks and cherries under the PDA model.
Note that the mean and the variance of pitchfork distributions under this model were also derived by~\citet[Proposition 5]{chang2010limit}.
Since the proof is similar to that in Corollary~\ref{cor:pda:ch}, we only outline the main steps used here.

\begin{proposition}
\label{prop:pda:distr}
For $n\geq 3$ we have 
\begin{align}
\ee _{\PDA}(\pf_n)&=\frac{n(n-1)(n-2)}{2(2n-3)(2n-5)} \sim\frac{n}{8}, \label{eq:pda:ch:mean}\\
{\cov_{\PDA}}(\pf_n,\ch_n)&=\frac{-n(n-1)(n-2)(n-3)}{2(2n-3)^2(2n-5)(2n-7)}
\sim -\frac{n}{32},\label{eq:pda:ch:pf:cov}\\
\var_{\PDA}(\pf_n)&=\frac{3n(n-1)(n-2)(n-3)(4n^3-40n^2+123n-110)}{4(2n-3)^2(2n-5)^2(2n-7)(2n-9)}
\sim
\frac{3n}{64}.
\label{eq:pda:pf:mean}
\end{align}

\end{proposition}

\begin{proof}
Applying Theorem~\ref{thm:pda:fun} to the function $\varphi(x,y)=x$ and using Corollary~\ref{cor:pda:ch}, we have
\begin{align*}
 \ee _{\PDA}(\pf_{n+1}) 
=\frac{3n(n-1)}{2(2n-1)(2n-3)}+\frac{2n-5}{2n-1}\ee _{\PDA}(A_n)
\end{align*}
for $n>2$. Now Eq.~(\ref{eq:pda:ch:mean}) follows by solving the last recursion with an approach similar to that in Corollary~\ref{cor:pda:ch}.\\

To this end, applying Theorem~\ref{thm:pda:fun} to the function $\varphi(x,y)=xy$ implies that
\begin{align*}
\ee _{\PDA}(\pf_{n+1}\ch_{n+1}) =\frac{2n-7}{2n-1} \ee _{\PDA} (\pf_n\ch_n)+\frac{n(n-1)(5n^2-9n-8)}{4(2n-1)(2n-3)(2n-5)} 
\end{align*}
holds for $n>2$. Solving this recursion we have\begin{align}
\label{eq:u:mix}
\ee _{\PDA}(\pf_n\ch_n)=\frac{n(n-1)(n-2)(n^2-3n-2)}{4(2n-3)(2n-5)(2n-7)},
\end{align}
from which Eq.~(\ref{eq:pda:ch:pf:cov}) follows. 

Finally, applying Theorem~\ref{thm:pda:fun} to the function $\varphi(x,y)=x^2$ shows that
\begin{align*}
\ee _{\PDA} (\pf^2_{n+1}) =\frac{2n-9}{2n-1} \ee _{\PDA} (\pf^2_n)+\frac{g(n+1)}{4(2n-1)(2n-3)(2n-5)(2n-7)} \\
\end{align*}
holds for $n>2$. 
Solving the above recursion leads to 
$$
\ee _{\PDA}\pf_{n}^2 =  \frac{n(n-1)(n-2)(n^3-4n^2-17n+66)}{4(2n-3)(2n-5)(2n-7)(2n-9)}, 
$$
from which Eq.~(\ref{eq:pda:pf:mean}) follows.
\epf
\end{proof}

We end this section with the following correlation result for the PDA model. 

\begin{corollary}
\label{cor:pda:cor}
For $n \geq 4$ we have
\begin{align}
\label{eq:corr:yule}
\corr _{\PDA}(\pf_n,\ch_n)=-\sqrt{\frac{2(2n-5)(2n-9)}{3(2n-7)(4n^3-40n^2+123n-110)}}
\sim -\frac{1}{\sqrt{3}\,n}. 
\end{align}
 In addition, $\{|\corr _{\PDA}(\pf_n,\ch_n)|\}_{n\geq 4}$ is a decreasing sequence converging to $0$.
\end{corollary}
\begin{proof}
Note first that Eq.~(\ref{eq:corr:yule}) follows from  Corollary~\ref{cor:pda:ch} and Proposition~\ref{prop:pda:distr}. Since the sequence $\{|\corr _{\PDA}(\pf_n,\ch_n)|\}_{n\geq 4}$ clearly approaches $0$, it remains to show that this sequence is decreasing. To this end, it suffices to show that the ratio
$$
R(n)=\frac{\corr _{\PDA}(\pf_n,\ch_n)^2}{\corr _{\PDA}(\pf_{n+1},\ch_{n+1})^2}
$$
is greater than 1 for $n\geq 4$. 
Using Eq.~(\ref{eq:corr:yule}), we have
$$
R(n)=\frac{(2n-5)^2(2n-9)(4(n+1)^3-40(n+1)^2+123(n+1)-110}{(2n-7)^2(2n-3)(4n^3-40n^2+123n-110)}.
$$
By numerical computation, we can check that $R(n)>1$ for $4\leq n \leq 15$, therefore we may assume that $n>15$ in the remainder of the proof. Now denoting the numerator and denominator of $R(n)$ by $R_1(n)$ and $R_2(n)$, respectively, then we have 
\begin{align*}
R_1(n)-R_2(n)&=64n^5-944n^4+5408n^3-15048n^2+20436n-10995 \\
&> 64n^4(n-15)+5408n^3(n-15)+20436(n-15) 
> 0
\end{align*}
for $n>15$. This implies $R(n)=R_1(n)/R_2(n)>1$ for $n>15$, as required. 
\epf
\end{proof}

\section{A comparative study of two models}
\label{sec:comp}

In this section, we compare and contrast the distributional properties of the number of cherries and the pitchforks in random trees generated under the YHK and the PDA models. 

To begin with, note that the recursions in Theorems~\ref{thm:yule:pf} and~\ref{thm:pda:pf} provide us exact computation of the joint distribution, and hence also the marginal distributions, of $\pf_n$ and $\ch_n$ under the two models.  For example, the joint distributions with $n=200$ for the two models are depicted in Fig.~\ref{fig:jd}, which suggests the joint distributions under both models can be well-approximated by bivariate normal distributions.
In addition, when $n=200$, the maximum for the joint distribution under the YHK model is $0.177...$, which is attained at $(\pf_n,\ch_n)=(33,67)$. On the other hand, under the PDA model the maximum is $0.145...$, which is attained at $(\pf_n,\ch_n)=(25,50)$. 
On average, the result below shows that trees generated by the YHK model contain more cherries and more pitchforks than trees of the same size generated by the PDA model.


\begin{proposition}
\label{prop:ch:comp}
For $n> 3$, we have
\begin{align}
\label{eq:ch:mean:comp}
\ee_{\PDA}(C_n)<\ee_{\YHK}(C_n) <\frac{4}{3}\ee_{\PDA}(C_n)
\end{align}
and
\begin{align}
\label{eq:pf:mean:comp}
\ee _{\PDA}(A_n)<\ee _{\YHK}(A_n)<\frac{4}{3}\ee _{\PDA}(A_n).
\end{align}
\end{proposition}

\begin{proof}
By Corollaries~\ref{cor:yule:ch} and~\ref{cor:pda:ch} we have
$$
\ee_{\YHK}(C_n) = \left[1 + \frac{n-3}{3(n-1)} \right ] \ee_{\PDA}(C_n), 
$$
from which Eq.~(\ref{eq:ch:mean:comp}) follows.
Similarly, by Propositions~\ref{prop:yule:distr} and~\ref{prop:pda:distr} we have 
$$
\ee _{\YHK}(A_n) =\left[1+ \frac{n^2-7n +9}{3(n-1)(n-2)}\right] \ee _{\PDA}(A_n),
$$
from which Eq.~(\ref{eq:pf:mean:comp}) follows. 
\epf
\end{proof}

Next, we study the variances of cherry and pitchfork distributions under the two models. 

\begin{proposition}
For $n> 5$, we have 
$$\frac{32}{45}\var _{\YHK}(C_n)<\var _{\PDA}(C_n)< \frac{49}{54} \var _{\YHK}(C_n).$$
\end{proposition}

\begin{proof}

Let $R_n = \var _{\YHK}(C_n)/\var _{\PDA}(C_n)$. Then by Corollaries~\ref{cor:yule:ch} and~\ref{cor:pda:ch} we have
$$
R_n = \frac{4(2n-3)^2(2n-5)}{45(n-1)(n-2)(n-3)}. $$
This implies
$$\frac{R_n}{R_{n+1}} = \frac{n(2n-3)(2n-5)}{(n-3)(2n-1)^2} = 1 + \frac{2n+3}{(n-3)(2n-1)^2}>1,$$ 
and hence that $R_n$ is decreasing in $n$. Noting that 
$\lim_{n\to \infty} R_n = \frac{32}{45}$, we have 
$$ \frac{32}{45} < R_n < R_5 = \frac{49}{54} <1,
$$
from which the proposition follows. 
\epf
\end{proof}


\begin{proposition}
\label{prop:pf:comp}
For $n\geq 7$, we have
$$
1.168\,{\var _{\PDA}(A_n)}<\frac{368}{315}{\var _{\PDA}(A_n)}<\var _{\YHK}(A_n)<\frac{8349}{6520}{\var _{\PDA}(A_n)}<1.281\, {\var _{\PDA}(A_n)} .
$$
\end{proposition}

\begin{proof}
Let $R_n = \var _{\YHK}(A_n)/\var _{\PDA}(A_n) $ for $n\geq 7$. Then by Proposition~\ref{prop:yule:distr} and~\ref{prop:pda:distr} we have 
$$ R_n = \frac{23(2n-3)^2(2n-5)^2(2n-7)(2n-9)}{315(n-1)(n-2)(n-3)(4n^3-40n^2+123n -110)}, $$
and hence
\begin{align*}
\frac{R_n}{R_{n+1}} &= \frac{n(2n-5)(2n-9)(4n^3-28n^2+55n-23)}{(n-3)(2n-1)^2(4n^3-40n^2+123n-110)} \\
&= 1+  \frac{2(24n^3-180n^2+382n-165)}{(n-3)(2n-1)^2(4n^3-40n^2+123n-110)}>1.
\end{align*}
Therefore,  $R_n$ is strictly decreasing in $n$ and we have
$$
\frac{\var _{\YHK}(A_n)}{\var _{\PDA}(A_n)} 
=R_n > \lim_{m \to \infty} R_m 
= \frac{23 \cdot 4 \cdot 4\cdot 4 }{315 \cdot 4} =\frac{368}{315}>1.168\,.
$$
This, together with $R_7=\frac{8349}{6520}<1.281$, completes the proof.
\epf
\end{proof}

Proposition~\ref{prop:ch:comp} shows that trees generated by the YHK model have smaller variation in the number of cherries than trees of the same size generated by the PDA model. On the contrary, Proposition~\ref{prop:pf:comp} shows that YHK model generates trees with larger variation in the number of pitchforks than the PDA model does. This is not unexpected as the covariances of cherries and pitchforks are found to be negative by Propositions~\ref{prop:yule:distr} and~\ref{prop:pda:distr}.

Now we present a result concerning the correlation coefficients  between the cherry and the pitchfork distributions under the two models. Since these two distributions are negatively correlated, we will focus on their absolute values. 

\begin{proposition}
For $n\geq 7$, we have
$|\corr _{\YHK}(A_n, C_n)|\geq {|\corr _{\PDA}(A_n, C_n)|}.
$ 
Moreover, $\big\{\frac{\corr _{\PDA}(A_n, C_n)}{\corr _{\YHK}(A_n, C_n)} \big\}_{n\ge 7}$ is a monotonically decreasing sequence with limit $0$.
\end{proposition}

\begin{proof}
This follows from Corollary~\ref{cor:yule:cor} and~\ref{cor:pda:cor} and the observation that \\
$$
|{\corr _{\PDA}}(A_7,C_7)|=\sqrt{\frac{30}{7\times 163}}\leq {|\corr _{\YHK}(A_7,C_7)|}= 
\sqrt{\frac{14}{69}}.
$$
\epf
\end{proof}

\begin{figure}[h]
\begin{center}
{\includegraphics[scale=0.5]{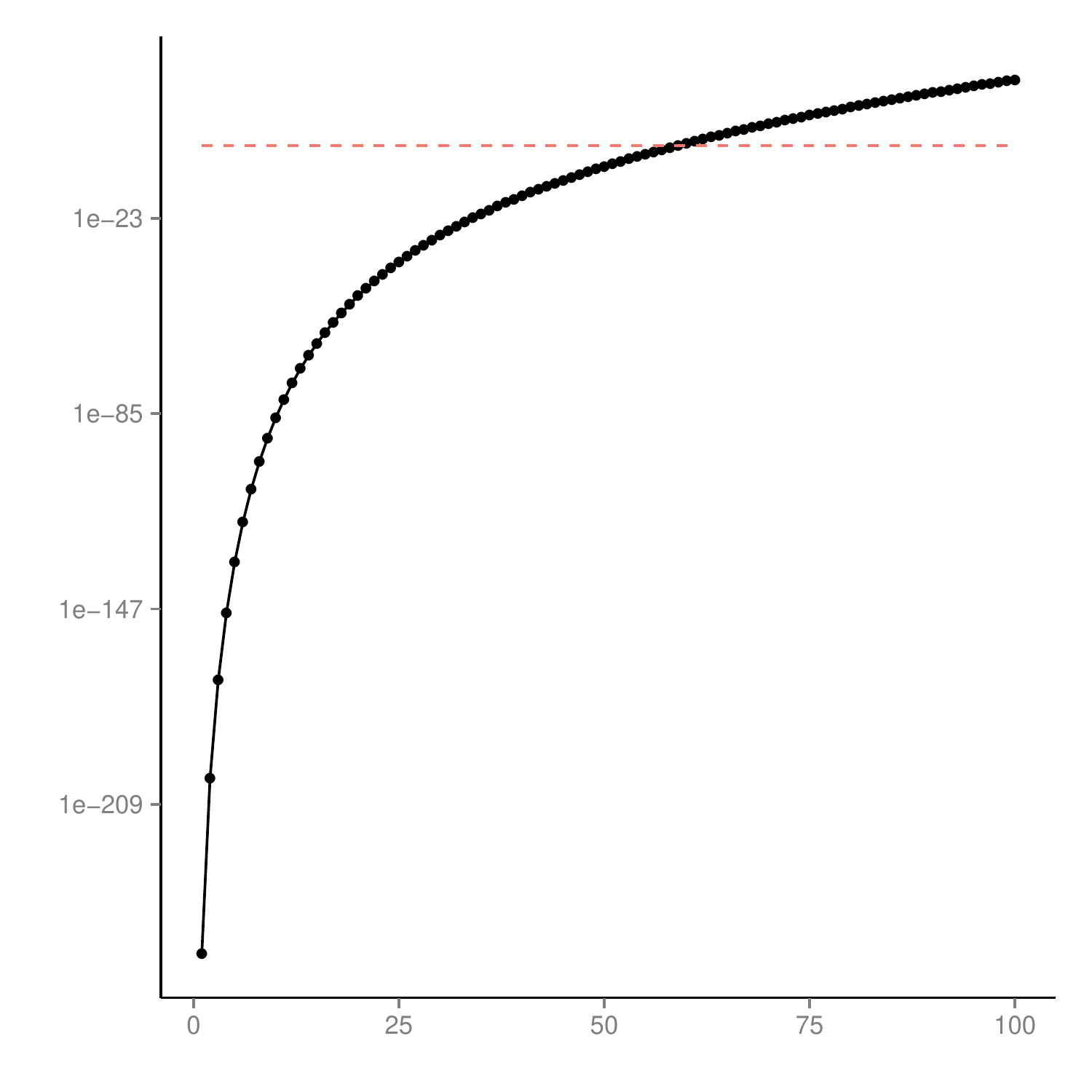}}
\end{center}
\caption{Plot of the ratio ${\pyule(C_n=k)}/{\puni(C_n=k)}$ for $n=200$ and $1\le k \le 100$. Here the $y$-axis is in the logarithmic scale; the dotted line $y=0$ indicates the ratio being $1$.
}
\label{fig:ratio}
\end{figure} 

\bigskip

Proposition~\ref{prop:ch:comp} states that the mean of $\ch_n$ is greater under the YHK model  than the PDA model. 
In the remainder of this section, we shall present a more detailed study on $\ch_n$. Intuitively, it is easy to see that the number of cherries contained in a random tree generated by the YHK model is likely to be greater than that  by the PDA model: firstly, 
by Proposition~\ref{prop:type} we know that the number of cherries in $T[e]$, the phylogenetic tree obtained from $T$ by attaching a new leaf to edge $e$ in $T$, is strictly greater than that in $T$ precisely when $e$ is a pendant edge of $T$;
 secondly,  in the YHK process the edge to which the new leaf is attached is   sampled only from  the pendant edges while in the PDA model that edge is sampled from all possible edges. Indeed, this intuition can also be corroborated  by numerical results.  As an example, considering the ratio of 
$\pyule(C_n=k)/\puni(C_n=k)$ with $n=200$ as depicted in Fig.~\ref{fig:ratio} using a logarithmic scale, then it is clear that the ratio is strictly increasing and is greater than 1 when $k$ is greater than a certain value. 
The following theorem establishes the existence of a unique change point between the two models for $n\geq 4$. Note that a similar phenomenon is shown to hold for clade sizes by~\citet[Theorem 5]{zhu2014clades}.

\begin{theorem}
\label{thm:change}
Suppose $n\ge 3$. The ratio  $\pyule(C_n=k)/\puni(C_n=k)$ is strictly increasing  for $1\leq k \leq n/2$. In particular, there exists a number $\tau_n$ with $1\leq \tau_n\leq n/2$ such that 
$$
\pyule(C_n=k)< \puni(C_n=k)
~~~\mbox{for $1\le k< \tau_n$, and~~}
\pyule(C_n=k)> \puni(C_n=k)
~~~\mbox{for $\tau_n<k\le n/2$.}
$$
\end{theorem}

\begin{proof}
For simplicity, put $a^k_n=\pyule(C_n=k)$. By Eq.~(\ref{eq:ch:pda:inc}), it suffices to show that 
\begin{align}
\label{eq:ch:ratio}
f(k,n)=:\frac{a^{k-1}_n}{a^k_n} \leq \frac{\puni(C_n=k-1)}{\puni(C_n=k)}= \frac{4k(k-1)}{(n-2k+1)(n-2k+2)}:=g(k,n)
\end{align}
holds for $1<k\leq n/2$. To this end, we shall use induction on $n$. The base case $n=3$ is clear because  $\pyule(C_3=1)=\puni(C_3=1)=1$. 
For induction step, assuming that $f(k,m)< g(k,m)$ holds for a given $m>3$ and all $1\leq k \leq m/2$, it remains to show that $f(k,m+1)\leq g(k,m+1)$ for all $1< k\leq (m+1)/2$. 

Note first that $f(1,m+1)=g(1,m+1)=0$. Now fix $2\leq k\leq (m+1)/2$, then $a_m^{k-1}>0$. Since Proposition~\ref{prop:yule:ch:fun} implies
\begin{align*}
a_{m+1}^{k-1}=\frac{2k-2}{m}a_m^{k-1}+\frac{m-2k+4}{m}a_m^{k-2}
~~\mbox{and}~~
a_{m+1}^{k}=\frac{2k}{m}a_m^{k}+\frac{m-2k+2}{m}a_m^{k-1},
\end{align*}
it follows that Eq.~(\ref{eq:ch:ratio}) is equivalent to  
\begin{align*}
(2k-2)a_m^{k-1}+(m-2k+4)a_m^{k-2} < \big(2ka_m^k+(m-2k+2)a_m^{k-1}\big)\,g(k,m+1).
\end{align*}
Since $a_m^{k-1}>0$, dividing both sides of the last inequality by $a_m^{k-1}$ leads to 
\begin{align}
\label{eq:ch:ratio:alt}
(2k-2)+(m-2k+4)f(k-1,m) <  2k\frac{g(k,m+1)}{f(k,m)}+(m-2k+2)g(k,m+1).
\end{align}
By induction assumption we have $f(k-1,m)\leq g(k-1,m)$ and $0<f(k,m)< g(k,m)$, hence it remains to show that 
\begin{align}
\label{eq:ch:ratio:mid}
(2k-2)+(m-2k+4)g(k-1,m) <  2k\frac{g(k,m+1)}{g(k,m)}+(m-2k+2)g(k,m+1).
\end{align}
To this end, denoting the left  and the right side of Inequality~(\ref{eq:ch:ratio:mid}) by $L(k,m)$ and $R(k,m)$, respectively, then we need to show that $R(k,m)-L(k,m)>0$. Note first that
$$
L(k,m)=(2k-2)+\frac{4(k-1)(k-2)}{m-2k+3}.
$$
On the other hand, since 
$$
\frac{g(k,m+1)}{g(k,m)}=\frac{m-2k+1}{m-2k+3}=1-\frac{2}{m-2k+3},
$$
we have
$$
R(k,m)=2k-\frac{4k}{m-2k+3}+\frac{4k(k-1)}{m-2k+3}.
$$
Therefore, we have
\begin{align*}
R(k,m)-L(k,m) &=2+\frac{4k(k-1)-4k-4(k-1)(k-2)}{m-2k+3} \\
&=2+\frac{4(k-2)}{m-2k+3}>0,
\end{align*}
as required. Here the last inequality follows from $m>3$ and $2\leq k \leq (m+1)/2$.
\epf
\end{proof}

\section{Discussion and Conclusion}
\label{sec:discussion}

Tree shape indices are summary statistics of some aspect of the shape of a phylogenetic tree, particularly the `balance' of a tree. Since the introduction of the first tree shape index by~\citet{Sackin1972}, many such indices have been proposed (see~\citet{mooers97a} for an excellent review and ~\citet{mir2013new} for some recent development).

In this paper we present several results concerning the distributions of cherries and pitchforks under the YHK and PDA models. Our main results include two novel recursive formulae on the joint distributions of cherries and pitchforks under these two models, which enable us to numerically compute their joint probability density functions (and hence also the marginal distributions) for trees of any size numerically. This is relevant because one of the main applications of tree indices  is their use as test statistics to discriminate stochastic models of evolution~\citep[see e.g.][and the references therein]{blum2006random}, 
and we will pursue this application more thoroughly elsewhere.

\begin{figure}[h]
\begin{center}
{\includegraphics[scale=0.5]{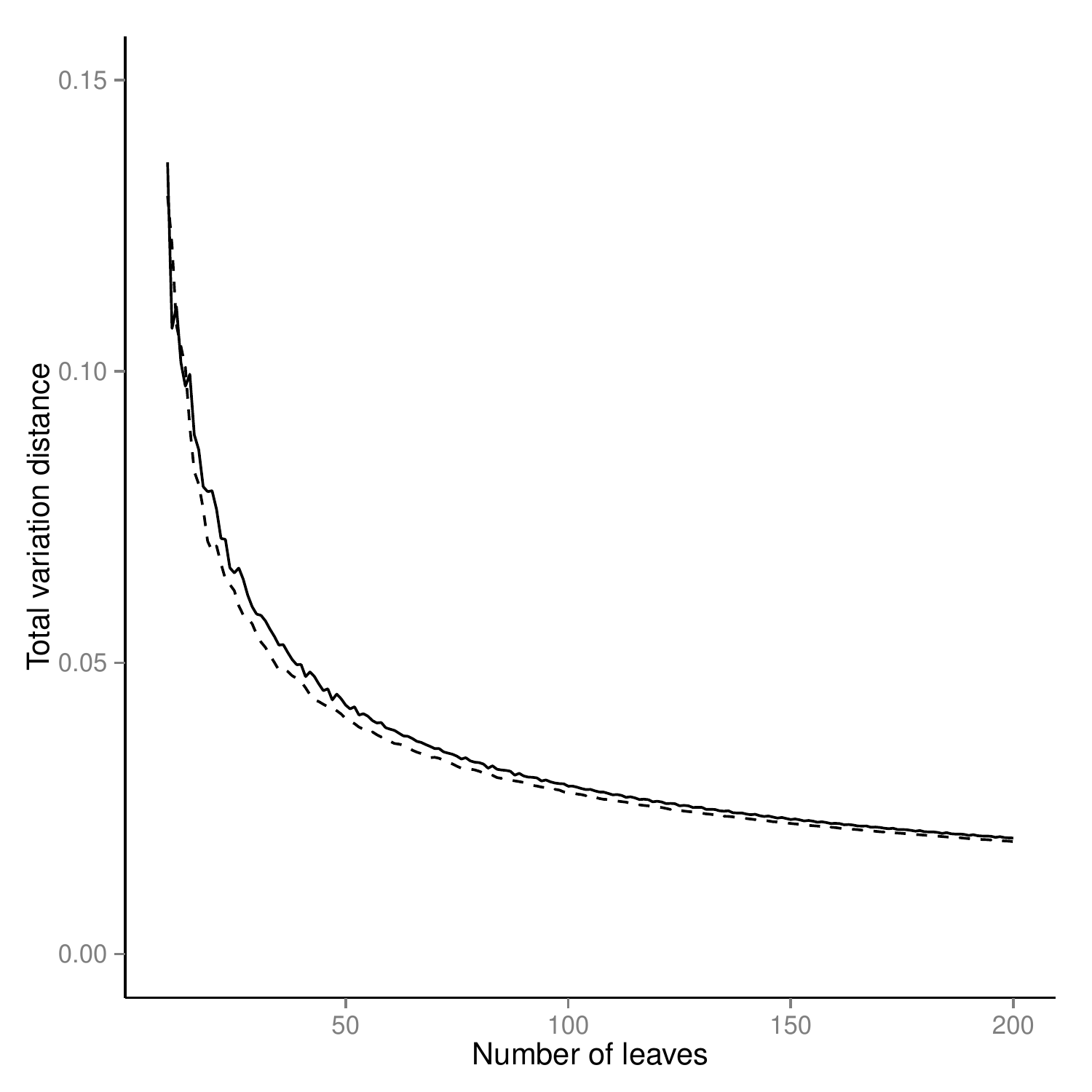}}
\end{center}
\caption{Plot of the total variation distance between the joint distributions of cherry and pitchfork and discretised bivariate normal distributions for $10\leq n \leq 200$ under the YHK model (solid line) and PDA model (dotted line). 
}
\label{fig:app}
\end{figure}

Our numerical results (e.g. Fig.~\ref{fig:app}) indicate 
 that the limiting joint distributions of cherries and pitchforks can be well approximated by bivariate normal distributions. For the YHK model, this was recently confirmed by~\citet{Janson2014}  and it remains open to establish the analogous result for the PDA model. 
In addition, it would be interesting to study the rate of the convergence. 

In this paper we concentrate on rooted trees, but it is of interest to investigate to what extent the results obtained in this paper for rooted trees can be carried over to unrooted trees. 
For example, using Eq.~(\ref{eq:chDistr:unrooted:pda}) and an argument similar to the proof in Theorem~\ref{thm:lc:pda}, it follows that the cherry distribution of unrooted trees under the PDA model is also log-concave, and hence unimodal. However, whether the same property holds for the cherry distribution of unrooted trees  under the YHK model remains open. One challenge is to derive the recursions for unrooted trees as in
Theorem~\ref{thm:yule:pf} or an exact formula as in Theorem~\ref{ch:pda}.

To end this article, we mention several additional questions. For instance, is pitchfork distribution, or other subtree distribution, log-concave? Our numerical calculation suggests the pitchfork distribution is log-concave. 
A related question is whether there also exists a unique change point
for other subtree distribution. Finally,  cherry pattern and pitchfork pattern are closely related to instances of recursive shape index (in the sense of~\citet{matsen2007optimization}), therefore it would also be of interest to see whether some of the properties obtained here can be carried over to some other tree indices as well.

\bigskip
\noindent
{\bf Acknowledgements} 
K.P. Choi acknowledges the support of Singapore Ministry of Education Academic Research Fund R-155-000-147-112.


 \newcommand{\noop}[1]{}

\end{document}